\documentclass[onecolumn]{article}

\usepackage[english]{babel}				
\usepackage{cite}						
\usepackage{amssymb}					
\usepackage[alsoload=binary]{siunitx}			
\usepackage[printonlyused,withpage]{acronym}	
\usepackage{etoolbox}					
\usepackage{xstring}					
\usepackage{amsthm} 					
\usepackage{mathtools}					
\usepackage{authblk}					

\newacro{sinr}[SINR]{signal-to-interference-plus-noise ratio}
\newacro{lte}[LTE]{Long Term Evolution}
\newacro{3gpp}[3GPP]{3rd Generation Partnership Project}
\newacro{comp}[CoMP]{Coordinated Multi-Point}
\newacro{bs}[BS]{base station}
\newacro{jt}[JT]{Joint Transmission}
\newacro{dps}[DPS]{Dynamic Point Selection}
\newacro{cs}[CS]{Coordinated Scheduling}
\newacro{eicic}[eICIC]{enhanced Inter-Cell Interference Cancellation}
\newacro{csi}[CSI]{channel state information}
\newacro{cran}[C-RAN]{Cloud-Radio Access Network}
\newacro{ue}[UE]{user equipment}
\newacro{ofdma}[OFDMA]{Orthogonal Frequency Division Multiple Access}
\newacro{prb}[PRB]{Physical Resource Block}
\newacro{siso}[SISO]{Single-Input-Single-Output}
\newacro{csirs}[CSI-RS]{\ac{csi} Reference Signal}
\newacro{pf}[PF]{Proportional Fair}
\newacro{mcs}[MCS]{modulation and coding scheme}
\newacro{fdd}[FDD]{Frequency Division Duplexing}
\newacro{crs}[CRS]{Cell-Specific Reference Signal}
\newacro{inlp}[INLP]{integer non-linear program}
\newacro{ilp}[ILP]{integer linear program}

\newcommand{\M}{\mathcal{M}}
\newcommand{\N}{\mathcal{N}}
\newcommand{\myL}{\mathcal{L}}
\newcommand{\I}{\mathcal{I}}
\newcommand{\J}{\mathcal{J}}
\newcommand{\note}[2]{\ac{#1}~${#2}$}

\newcommand*{\balfa}[2]{%
	\IfEqCase{#1}{
		{0}{\boldsymbol{\bar{\alpha}}}
		{1}{\boldsymbol{\bar{\alpha}}_{#2}}
	}[]%
}
\robustify{\balfa}

\newcommand{\mcsi}[1]{%
	\IfEqCase{#1}{
		{0}{\ac{csi}$_\text{lte}$}
		{8}{\ac{csi}$_\text{lte}^\text{R-8}$}
		{10}{\ac{csi}$_\text{lte}^\text{R-10}$}
		{11}{\ac{csi}$_\text{lte}^\text{R-11}$}
	}[]%
}
\robustify{\mcsi}

\newtheorem{ex}{Example}
\newtheorem{cor}{Corollary}
\newtheorem{prop}{Proposition}

\begin{document}

\title{Centralized Coordinated Scheduling in LTE-Advanced Networks}

\author[1]{Oscar~D.~Ramos-Cantor\thanks{Email: oramos@nt.tu-darmstadt.de}}
\author[2]{Jakob~Belschner}
\author[1]{Ganapati~Hegde}
\author[1]{Marius~Pesavento}
\affil[1]{Communication Systems Group, Darmstadt University of Technology, 64283 Darmstadt, Germany}
\affil[2]{Technology Innovation Division, Deutsche Telekom AG, 64295 Darmstadt, Germany}

\date{\vspace{-5ex}}
\maketitle

\begin{abstract}
This work addresses the problem associated with coordinating scheduling decisions among multiple base stations in an LTE-Advanced downlink network in order to manage inter-cell interference with a centralized controller. To solve the coordinated scheduling problem an integer non-linear program is formulated that, unlike most existing approaches, does not rely on exact channel state information but only makes use of the specific measurement reports defined in the 3GPP standard. An equivalent integer linear reformulation of the coordinated scheduling problem is proposed, which can be efficiently solved by commercial solvers. Extensive simulations of medium to large-size networks are carried out to analyze the performance of the proposed coordinated scheduling approaches, confirming available analytical results reporting fundamental limitations in the cooperation due to out-of-cluster interference. Nevertheless, the schemes proposed in this paper show important gains in average user throughput of the cell-edge users, especially in the case of heterogeneous networks.\\

\end{abstract}

\section{Introduction}\label{sec_intro}
Interference is one of the main limiting factors of today's cellular communication networks in terms of user and network throughputs, especially when operating with full frequency reuse to achieve high spectral efficiency \cite{Power_allocation_Kiani,power_control_coord_scheduling_Kiani, int_4G_Vrzic,int_align_Alcatel}. Nowadays, the demand for high data rates is constantly increasing \cite{Cisco_VNI_2016}. In modern cellular networks the users expect to enjoy excellent network performance irrespective of their geographic location and the load conditions of the network. Thus, new solutions are required in order to fulfill the ever increasing requirements, in particular for the users located at the cell-edge suffering from large path loss and strong inter-cell interference. Promising advances in this aspect have been made with multi-antenna technology \cite{SMUX_Tse,STC_Paulraj, STC_Gershman, mimo_lte_Lee}, network densification with interference management schemes \cite{HetNets_Andrews,survey_LTE_Lee,eICIC_HetNets_LopezPerez}, and \ac{comp} transceiver techniques \cite{CoMP_Droste,CoMP_Tx_Rx_Lee}.

In this work \ac{comp} network operation is studied, where the \acp{bs}, connected within a cooperation cluster, are prompted to cooperate with each other with the objective of improving the overall network performance, even at the expense of their individual cell or user throughputs \cite{cooperation_Larsson}. In the literature, three main \ac{comp} schemes are considered for the downlink scenario \cite{CoMP_Baracca,CoMP_Beylerian}. These are: \textit{i)} \ac{jt}, where multiple \acp{bs} simultaneously transmit a common message to a \ac{ue}, usually located at the cell-edge, \textit{ii)} \ac{dps}, where at each transmission time interval, the \ac{ue} can be served by a different \ac{bs} without triggering handover procedures, and \textit{iii)} \ac{cs}, where the \acp{bs} jointly make the scheduling decisions in order to manage the interference experienced by the \acp{ue} in the cooperation cluster \cite{SPAWC_Oscar}. This paper focuses on the last \ac{comp} scheme.

The performance of the above mentioned \ac{comp} schemes heavily depends on the \ac{csi} available at the transmitter. This \ac{csi} can be of different types such as  instantaneous channel coefficients or user's average achievable downlink data rates, among others, where the former represents the deepest level of detail and finest granularity, while the latter has the highest abstraction and aggregation levels. In practical downlink networks, where perfect global knowledge of the instantaneous channel coefficients is not available at the \acp{bs}, \ac{csi} is typically obtained in form of achievable data rate measurement reports generated by the \acp{ue}, averaged over multiple time/frequency/space dimensions and quantized to reduce the signaling overhead. Moreover, the \ac{csi} estimation process is only periodically carried out by the \acp{ue}, to limit the processing and transmission overheads, thus, saving energy consumption at the expense of outdated \ac{csi}. In this work, the \ac{comp} problem formulation is based on practical considerations of the \ac{csi}, in form of periodic achievable data rate measurement reports, in the following referred to as \emph{\ac{csi} reports}.

The network architecture, in which the \ac{comp} schemes are implemented, also influences the performance of such schemes. There are two main \ac{comp} network architectures, namely, centralized and decentralized \cite{central_decentral_CoMP}. In the case of centralized \ac{comp}, a central controller is connected to multiple \acp{bs} via backhaul links. This central controller is in charge of gathering and using the \ac{csi} reports, in order to make a coordinated decision among the connected \acp{bs}. For the decentralized \ac{comp} case, decisions are individually made by each \ac{bs} based on the information exchanged with neighboring \acp{bs}. A trade-off between coordination gains and system requirements, such as signaling overhead and computational complexity, needs to be found when designing a proper \ac{comp} solution. In the case of centralized \ac{comp}, high coordination gains are achievable at the expense of high computational complexity and large signaling overhead. On the other hand, decentralized \ac{comp} requires significantly less information exchange with lower coordination gains.

Over the past years, important research has been carried out regarding \ac{comp} schemes under different network architectures and \ac{csi} assumptions. In \cite{JT_DPS_Maattanen} and \cite{CoMP_DCS_JT_Mondal}, \ac{jt} and \ac{dps} schemes based on the enhanced \ac{csi} reports supported by \ac{lte}-Advanced Release 11, in the following denoted as \mcsi{11} reports, have been investigated. The results therein show throughput gains for the cell-edge users mainly, and the possibility to improve mobility management by means of \ac{dps}. Barbieri et al. studied \ac{cs} as a complement of \ac{eicic} in heterogeneous networks in  \cite{CoMP_HetNet_Barbieri}. In their scheme, cooperation takes place in form of \ac{cs} supported by beamforming in order to mitigate the interference caused by the macro \acp{bs}, to the \acp{ue} connected to the small cells. Multiple \ac{csi} reports are generated, where all possible precoders the macro \ac{bs} can select from a finite precoder codebook are considered for the cooperation. The results present negligible gain for \ac{eicic} with \ac{cs}, in comparison to \ac{eicic}-only.  In \cite{JT_C-RAN_Davydov}, a \ac{cran} architecture is used for centralized \ac{comp} \ac{jt} in heterogeneous networks, which enables the cooperation of larger cluster sizes. In that case, gains over \ac{eicic}-only are observed, especially for large cluster sizes. Authors in \cite{Nokia_2014} propose centralized and decentralized \ac{comp} \ac{cs} schemes that utilize \mcsi{11} reports, in which \emph{muting} is applied to one \ac{bs} at a time. A \ac{bs} is called muted if it does not transmit data on a specific time/frequency resource to any of its connected \acp{ue}. It has been shown that under this muting condition, both centralized and decentralized schemes achieve the same performance, favoring the decentralized scheme due to the reduced information exchange. Moreover, in \cite{Nokia_2015} the authors extend the cooperation scheme of \cite{Nokia_2014}, to introduce muting of more than one \ac{bs} per scheduling decision in a larger network. A greedy \ac{cs} algorithm is presented to solve the centralized problem, which yields limited additional gain with respect to the decentralized scheme with overlapping cooperation clusters. The coordination scheme of \cite{Nokia_2015} consists in a greedy optimization procedure. It is therefore suboptimal and further investigation regarding the optimally achievable performance of coordination, in the case of \mcsi{11} reports, has not been carried out. Additionally, the results are focused on macro-only networks, where the gains of cooperation are restricted due to similar interfering power levels experienced from multiple \acp{bs}.

Although the above mentioned works show that \ac{comp} schemes enhance the user throughput with respect to a network operating without any cooperation, no detailed studies are carried out in order to establish the maximum achievable gains that \ac{comp} schemes can offer in realistic network scenarios and under \ac{lte}-Advanced specific \ac{csi} reports. In \cite{Comp_limit_Lozano}, it has been demonstrated from an analytical perspective that cooperative schemes have fundamentally limited gains. That is, even under the assumption of centralized coordination and ideal \ac{csi} in form of instantaneous channel coefficients, the cooperation gains are limited due to the residual interference from \acp{bs} outside of the cooperation area, the signaling overheads and the finite nature of the time/frequency/space resources. In the paper at hand, such limits are investigated under practical conditions by means of system level evaluations. For that purpose, an optimal \ac{comp} \ac{cs} scheme is proposed and analyzed in detail for an \ac{lte}-Advanced downlink network with centralized architecture. The problem formulation is based on multiple \ac{csi} reports generated by the \acp{ue} and gathered by the central controller, which uses this information to determine the coordinated scheduling decisions for all connected \acp{bs}. The central controller then decides which \acp{bs} serve their connected \acp{ue} on a given time/frequency resource, and which \acp{bs} are muted in order to reduce the interference caused to the \acp{ue} served by the neighboring transmitting \acp{bs}. The main contributions of this work are summarized as follows:
\begin{itemize}
	\item The \ac{cs} problem, where \acp{bs} cooperate by muting time/frequency resources  based on standardized \mcsi{11} reports, is formulated as an \ac{inlp}.
	\item The non-linear \ac{cs} with muting problem is reformulated into a computationally tractable \emph{equivalent} \ac{ilp}, which enjoys of low computational complexity and can be efficiently solved by commercial solvers. This reformulation is based on lifting technique and exploits specific separability and reducibility properties of the problem. Thus, making the optimization scheme applicable as a valuable benchmark scheme in middle to large scale networks.
	\item A configurable heuristic algorithm is proposed as an extension to the greedy algorithm in \cite{Nokia_2015}, which achieves an excellent trade-off between performance and computational complexity.
	\item Extensive numerical simulations are carried out, under practical scenarios for macro-only and heterogeneous networks, in order to assess the maximum achievable gains of the proposed and current \ac{comp} \ac{cs} schemes.
\end{itemize}

\section{System model}\label{sec_sys_mod}
A cellular network is considered as illustrated in \figurename~\ref{fig_sys_model}, where a cooperation cluster of $M$ \acp{bs}, operating in \ac{fdd} mode, serves $N$ \acp{ue} in the downlink. \ac{ofdma} is assumed with frequency reuse one, where at each transmission time interval, all \acp{bs} can make use of the same $L$ \acp{prb} for transmission. Thus, inter-cell interference affects the \acp{ue}, especially at the cell-edge. Additionally, interference from \acp{bs} outside of the cooperation cluster is considered. The operation of the cooperation cluster is managed by a central controller with backhaul connectivity to all $M$ \acp{bs}. In the following, the sets of indexes $\M = \{1,\dots,M\}$, $\N = \{1,\dots,N\}$ and $\myL = \{1,\dots,L\}$ are used to address the \acp{bs}, \acp{ue} and \acp{prb}, respectively.

The received power at \note{ue}{n\in\N}, from \note{bs}{m\in\M}, on \note{prb}{l\in\myL}, is denoted as $p_{n,m,l}$. Hence, for \ac{siso} transmission,
\begin{equation}
\label{eq_p_siso}
	p_{n,m,l}=|g_n\,h_{n,m,l}|^2\,\phi_{m,l},
\end{equation}
where $\phi_{m,l}$ corresponds to the transmit power of \note{bs}{m} on \note{prb}{l}, the complex coefficient $h_{n,m,l}$ represents the amplitude gain of the downlink channel between \note{bs}{m} and \note{ue}{n} on \note{prb}{l}, and $g_n$ is the receiver amplitude processing gain. In \eqref{eq_p_siso}, the transmitted symbols are assumed to exhibit unit average transmit power. When summing over all \acp{prb}, the total received power at \note{ue}{n} from \note{bs}{m} is obtained as
\begin{equation}
\label{eq_p_total}
	p_{n,m}=\sum_{l=1}^L p_{n,m,l}.
\end{equation}

\begin{figure}[!t]
	\centering
	\includegraphics{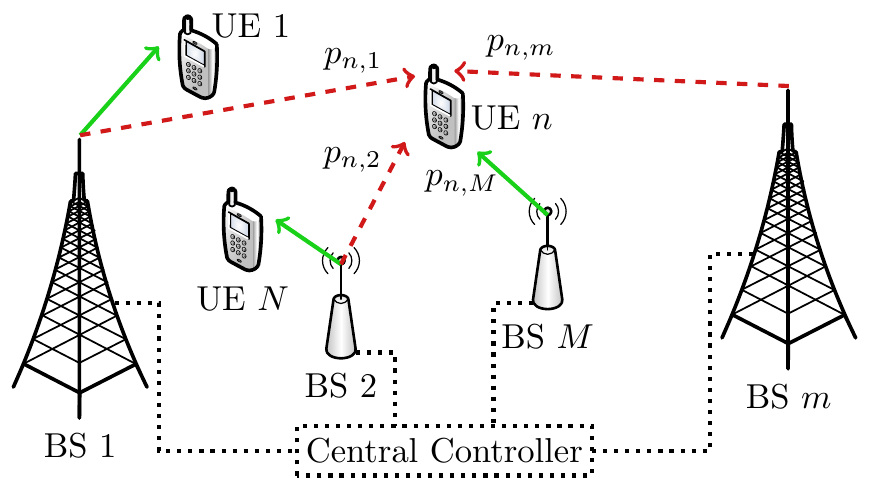}
\caption{Cooperation cluster of $M$ \acp{bs} and $N$ \acp{ue} in the downlink. The \acp{bs} are connected through the backhaul to a central controller.}
\label{fig_sys_model}
\end{figure}

Generally, the serving \ac{bs} of \note{ue}{n\in\N} is selected as the \ac{bs} from which the highest total received power is obtained, as defined in \eqref{eq_p_total}. In the case of heterogeneous networks, however, such a selection strategy causes a loss in the cell-splitting gains expected by the introduction of the small cells \cite{cell_split_gain_Galiotto}. In order to achieve load balancing, i.e., a ``fair" distribution of the served \acp{ue} among the small cells and the macro \acp{bs}, techniques like cell range expansion are applied in e.g., \ac{lte}-Advanced \cite{user_association_BIAS_Ye}, where the \acp{ue} are instructed to add a constant off-set in the computation of the total received power of the small cells. Thus, increasing the number of \acp{ue} served by the small cells. In this paper, both homogeneous macro-only and heterogeneous networks are studied, where for the latter case, cell range expansion is applied. The $N\times M$ connection matrix $\mathbf{C}$ is defined, with elements
\begin{equation}
\label{eq_c}
	c_{n,m}= \left\{
		\begin{aligned}
			1 & \ \text{if \ac{ue} }n \text{ is served by \ac{bs} }m\in\M\\
			0 & \ \text{otherwise},
		\end{aligned}\right.
\end{equation}
characterizing the serving conditions between \acp{bs} and \acp{ue}. It is assumed that only one \ac{bs} serves \note{ue}{n} over all \acp{prb}. Hence,
\begin{equation}
\label{eq_one_BS}
	\sum_{m\in\M} c_{n,m}=1 \ \forall n\in\N.
\end{equation}
 It is further assumed, for simplicity, that the \acp{ue} are quasi-static, such that no handover procedures are triggered between the \acp{bs}. Therefore, the connection matrix $\mathbf{C}$ is assumed to be constant during the considered operation time.

The set of indexes of \acp{bs} within the cooperation cluster, that interfere with \note{ue}{n\in\N} is defined as ${\I_n=\{m~|~c_{n,m}=0,~\forall m\in\M\}}$, with cardinality ${|\I_n|=M-1}$. Moreover, since \note{ue}{n} experiences different interfering power levels from the $\I_n$ interfering \acp{bs}, the set $\I_n' \subseteq \I_n$ of indexes of the $M'$ \emph{strongest interfering \acp{bs}} of \note{ue}{n}, is defined such that
\begin{equation}
\label{eq_In_prime}
\min_{m'\in\I_n'}p_{n,m'}\geq\max_{m\in\I_n\backslash\I_n'}p_{n,m},
\end{equation}
i.e., the set $\I_n'$ contains the indexes of the $M'$ interfering \acp{bs} with the highest total received power at \note{ue}{n}, as calculated in \eqref{eq_p_total}. The number of strongest interfering \acp{bs} is bounded as ${0 \leq M' \leq M-1}$, where equality of the sets $\I_n'$ and $\I_n$ holds, if $M'=M-1$. The sets $\I_n$ and $\I_n'$ apply for all \acp{prb} in the reporting period.

Within the cooperation cluster, the \acp{bs} cooperate in the form of coordinated scheduling with muting, as previously mentioned in Section~\ref{sec_intro}. The central controller is then, in charge of managing the downlink transmissions of the \acp{bs} where, at each transmission time interval and on a per \ac{prb} basis, each \ac{bs} can be requested to abstain from transmitting data. Hence, the interference caused to \acp{ue} located in neighboring \acp{bs} is reduced on the \acp{prb} with muted \acp{bs}. Given the muting decision matrix, $\balfa{0}{}$, of dimensions ${M\times L}$ and elements
\begin{equation}
\label{eq_mute_dec}
	\bar{\alpha}_{m,l}= \left\{
		\begin{aligned}
			1 & \ \text{if \ac{bs} }m\in\M \text{ is muted on \ac{prb} }l\in\myL\\
			0 & \ \text{otherwise},\\
		\end{aligned}\right.
\end{equation}
the \ac{sinr} of \note{ue}{n\in\N}, which is served by \note{bs}{k\in\M},  on \note{prb}{l}, is then defined as
\begin{equation}
\label{eq_sinr}
	\gamma_{n,l}\left(\balfa{1}{l}\right)=\frac{(1-\bar{\alpha}_{k,l})\,p_{n,k,l}}{I_{n,l}^{\text{cc}}\left(\balfa{1}{l}\right)+I_{n,l}^{\text{oc}}+\sigma^2}.
\end{equation}
The vector $\balfa{1}{l}$, is equivalent to the $l$-\textit{th} column of $\balfa{0}{}$. The numerator corresponds to the average received power at \note{ue}{n}, from the serving \note{bs}{k}, on \note{prb}{l}, as defined in \eqref{eq_p_siso}. The first term in the denominator corresponds to the average inter-cell interference from the \acp{bs} within the cooperation cluster, with
\begin{equation}
\label{eq_sinr_icc}
	I_{n,l}^{\text{cc}}\left(\balfa{1}{l}\right)=\sum_{m\in \I_n} (1-\bar{\alpha}_{m,l})\,p_{n,m,l},
\end{equation}
$I_{n,l}^{\text{oc}}$ is the average out-of-cluster interference, and $\sigma^2$ is the noise power assumed, without loss of generality, to be constant for all \acp{ue} over all \acp{prb}. It is worth to notice that the out-of-cluster interference $I_{n,l}^{\text{oc}}$ is assumed to be independent of the muting decision matrix $\balfa{0}{}$, since the central controller is not aware of the muting decisions made by the \acp{bs} outside of the cooperation cluster.

The achievable data rate of \note{ue}{n\in\N} on \note{prb}{l\in\myL}, is modeled as a function of the \ac{ue}'s \ac{sinr}. Hence,
\begin{equation}
\label{eq_r_gen}
	r_{n,l}=f\left(\gamma_{n,l}\right),
\end{equation}
where $f\left(\gamma_{n,l}\right)$ denotes a mapping function from the \ac{sinr} of \note{ue}{n} on \note{prb}{l}, to the achievable data rate.
\begin{ex}
A mapping function based on Shannon's capacity bound is
\begin{equation}
\label{eq_r_ex}
	f\left(\gamma_{n,l}\right)=D\log_2\left(1+\gamma_{n,l}\right),
\end{equation}
where the scaling factor $D$ is a general term to represent the \ac{prb} bandwidth and the selected \ac{mcs} \cite{CQI_table_Fan}. In \eqref{eq_r_ex}, it is assumed without loss of generality that all \acp{prb} correspond to the same time-bandwidth product.
\end{ex}

In the literature, it is common to assume that the central controller has perfect \ac{csi} knowledge, in form of the instantaneous channel coefficients, $h_{n,m,l}$, as introduced in \eqref{eq_p_siso}. Thus, the computation of the \ac{sinr}, $\gamma_{n,l}\left(\balfa{1}{l}\right)$, and the achievable data rates, $r_{n,l}$, as defined in \eqref{eq_sinr} and \eqref{eq_r_gen}, respectively, is carried out in a straightforward manner for any possible muting decision $\balfa{0}{}$. However, in practical conditions such as in \ac{lte}-Advanced networks, acquiring \ac{csi} in form of instantaneous channel coefficients represents a significant challenge, due to the large signaling overhead caused by the fine granularity of the time/frequency/space dimensions, and the high sensitivity to pilot contamination of the channel estimation, among others. For that reason, the \ac{csi} in \ac{lte}-Advanced is typically available in form of achievable data rate measurement reports, i.e., \ac{csi} reports, which contain average information of multiple time/frequency/space resources for a subset of possible muting decisions $\balfa{0}{}$, as defined in \eqref{eq_mute_dec}. Thus, the processing and signaling overhead is reduced, at the expense of limited \ac{csi} knowledge for the \ac{comp} \ac{cs} scheme.

\section{\ac{csi} reporting for \ac{lte}-Advanced \ac{comp} \ac{cs}}\label{sec_csi_rep}
To enable opportunistic scheduling in \ac{lte}, \ac{csi} reports are supported since the first release, i.e., Release~8 \cite{3GPP_TS_36.213}. In \ac{lte}-Advanced, \ac{comp} operation has been included in Release~11 and beyond. To support \ac{csi} estimation at the \acp{ue}, the transmission of \acp{csirs} from the \acp{bs}, has been introduced as an extension to the common \acp{crs} \cite{3GPP_TR_36.814,4G_Dahlman}. One major enhancement in \ac{lte}-Advanced, with respect to the \ac{lte} Release 8 \acp{crs}, is the possibility to configure muted \acp{csirs}, i.e., \acp{csirs} with zero transmission power, enabling the \acp{ue} to estimate \ac{csi} from specific neighboring \acp{bs} without interference from the serving \ac{bs}. Due to this feature, the \acp{ue} can generate multiple \mcsi{11} reports that reflect different serving and interfering conditions in the network \cite{JT_DPS_Maattanen,CoMP_DCS_JT_Mondal}.
\begin{ex}
In a specific \ac{csirs} configuration, \note{bs}{m\in\M} can be defined as muted, so that the \mcsi{11} report provides information regarding the achievable data rates for \ac{cs} with muting. The specific muting decision, considered at the central controller, is $\balfa{0}{}=\mathbf{0}_{M\backslash m\times L}$, where $\mathbf{0}_{M\backslash m\times L}$ is an ${M\times L}$ matrix with zero elements in all but the $m$-\textit{th} row.
\end{ex}

For a cooperation cluster with $M$ \acp{bs}, a total of ${J=2^M-1}$ muting decisions can be made per \note{prb}{l\in\myL}. In this work the practical case is considered, where the \ac{cs} operation is managed by the central controller based on the \mcsi{11} reports provided by the \acp{ue}. Since the \acp{sinr} and achievable data rates of \note{ue}{n\in\N}, as defined in \eqref{eq_sinr} and \eqref{eq_r_gen}, respectively, are dominated by its strongest interfering \acp{bs} \cite{Nokia_2014}, in the following it is assumed that the \acp{ue} generate a total of $J'=2^{M'}$ \mcsi{11} reports per \ac{prb}, with $J'<J$. Then, each \mcsi{11} report only considers the $M'$ strongest interfering \acp{bs} of \note{ue}{n}, as described by the set $\I_n'$, introduced in Section~\ref{sec_sys_mod}.

Each of the $J'$ \mcsi{11} reports, generated by \note{ue}{n\in\N} on \note{prb}{l\in\myL}, reflects a unique \emph{interference scenario} for its strongest interfering \acp{bs}. More specifically, the interference scenario ${j\in\J'}$, with ${\J'=\{1,\dots,J'\}}$, is characterized by the \emph{muting indicator} set $\J_{n,j}$, which contains the indexes of the (strongest) interfering \acp{bs} considered to be muted in the $j$-\textit{th} \mcsi{11} report of \note{ue}{n}. Hence, the set $\J_n=\mathbb{P}(\I_n')$ contains all $J'$ muting indicator sets for \note{ue}{n}, with $\mathbb{P}(\cdot)$ denoting the set of all subsets of $\I_n'$. The set $\J_n$ is common to all \acp{prb}, due to the definition of the strongest interfering \acp{bs} of \note{ue}{n}, as in~\eqref{eq_In_prime}. From the muting indicator set, $\J_{n,j}$, the \emph{muting pattern} of the $j$-\textit{th} \mcsi{11} report of \note{ue}{n}, on \note{prb}{l}, is defined as
\begin{equation}
\label{eq_alpha}
	\alpha_{n,m'\!,l,j}= \left\{
		\begin{aligned}
			1 & \ \text{if } m'\!\in \J_{n,j} \text{ on \ac{prb} }l\\
			0 & \ \text{otherwise},\\
		\end{aligned}\right. \ \forall m'\!\in\I_n',
\end{equation}
i.e., $\alpha_{n,m',l,j}=1$, if the (strongest) interfering \note{bs}{m'}, is muted on \note{prb}{l}, under interference scenario ${j\in\J'}$. The definition in \eqref{eq_alpha} considers only the set of strongest interfering \acp{bs} of \note{ue}{n}, i.e., $\I_n'$. Therefore, a constant muting state of the remaining \acp{bs} in the cooperation cluster is required, for all $\J'$ interference scenarios.
In the following, it is assumed without loss of generality, that  $\alpha_{n,m,l,j}=0, \forall m\notin\I_n', \forall j\in\J'$. Although, the definition of the muting pattern in \eqref{eq_alpha} is similar to the definition of the muting decision in \eqref{eq_mute_dec}, the two concepts are different. The muting pattern describes the assumed muting conditions during the generation of the \mcsi{11} reports for the different interference scenarios, while the muting decision is imposed by the central controller, to the \acp{bs} within the cooperation cluster, as the result of the implementation of the \ac{cs} with muting scheme.

For the generation of the \mcsi{11} reports, \note{ue}{n\in\N} calculates the \ac{sinr} and the achievable data rates, on \note{prb}{{l\in\myL}} under interference scenario ${j\in\J'}$. Therefore, similar to \eqref{eq_sinr}, the \ac{sinr} of \note{ue}{n} on \note{prb}{l}, under interference scenario $j$, is defined as
\begin{equation}
\label{eq_sinr_2}
	\gamma_{n,l,j}\left(\alpha_{n,m'\!,l,j}\right)=\frac{p_{n,k,l}}{I_{n,l,j}^{\text{si}}\left(\alpha_{n,m'\!,l,j}\right)+I_{n,l}^{\text{wi}}+I_{n,l}^{\text{oc}}+\sigma^2},
\end{equation}
where the first term in the denominator of \eqref{eq_sinr} has been decomposed into two terms corresponding to the interference from the strongest interfering \acp{bs} of \note{ue}{n}, i.e., $I_{n,l,j}^{\text{si}}\left(\alpha_{n,m'\!,l,j}\right)$, and the interference from the remaining (weakest) interfering \acp{bs} of \note{ue}{n}, denoted by $I_{n,l}^{\text{wi}}$. From the previous discussion on the muting patterns of \note{ue}{n}, the interference from the strongest interfering \acp{bs} that can cooperate to improve the \ac{sinr} of \note{ue}{n}, depends on interference scenario~$j$, and thus, the muting pattern, as
\begin{equation}
\label{eq_si_int}
	I_{n,l,j}^{\text{si}}\left(\alpha_{n,m'\!,l,j}\right) = \sum_{m'\!\in \I_n'} (1-\alpha_{n,m'\!,l,j})\,p_{n,m'\!,l}.
\end{equation}
On the other hand, the interference from the weakest interfering \acp{bs} of \note{ue}{n} is assumed to be constant and independent of the possible muting decisions, with
\begin{equation}
\label{eq_wi_int}
	I_{n,l}^{\text{wi}} = \sum_{\mathclap{m\in \I_n\backslash\I_n'}} p_{n,m,l}.
\end{equation}
Furthermore, the out-of-cluster interference and the noise variance are also assumed to be constant terms among all the $J'$ interfering scenarios considered in the \mcsi{11} reports.

To complete the information for the \mcsi{11} reports, $r_{n,l,j}$ denotes the achievable data rate of \note{ue}{n\in\N}, on \note{prb}{l\in\myL}, considered under interference scenario~${j\in\J'}$. The calculation of $r_{n,l,j}$ follows the definition in \eqref{eq_r_gen}, with 
\begin{equation}
\label{eq_r_gen_2}
	r_{n,l,j}=f\left(\gamma_{n,l,j}\right).
\end{equation}

\begin{prop}\label{lem_sinr}
For \note{ue}{n\in\N}, on \note{prb}{l\in\myL}, if ${\J_{n,i}\subsetneq\J_{n,j}, \forall i,j\in\J', i\neq j}$, then $\gamma_{n,l,i}<\gamma_{n,l,j}$.
\end{prop}
\begin{proof}
See Appendix \ref{app_lem_sinr}.
\end{proof}
That is, the Proposition~\ref{lem_sinr} states that the \ac{sinr} of \note{ue}{n\in\N}, on \note{prb}{l\in\myL}, increases when muting additional (strongest) interfering \acp{bs}.

\begin{cor}\label{lem_rate}
For \note{ue}{n\in\N}, on \note{prb}{l\in\myL}, if ${\J_{n,i}\subsetneq\J_{n,j}, \forall i,j\in\J', i\neq j}$ and $f(\gamma_{n,l})$, introduced in \eqref{eq_r_gen}, is a non-decreasing function, then $r_{n,l,i}\leq r_{n,l,j}$.
\end{cor}

Hence, based on the Corollary \ref{lem_rate}, the achievable data rate of \note{ue}{n\in\N}, on \note{prb}{l\in\myL}, increases or remains constant when muting additional (strongest) interfering \acp{bs}. The observations in the Proposition~\ref{lem_sinr} and the Corollary~\ref{lem_rate} dictate the solution of the \ac{cs} with muting problem formulated in Section \ref{sec_cs}.

\begin{ex}\label{ex_csi}
In an exemplary network with a cooperation cluster of $M=4$ \acp{bs}, and a total of $M'=2$ strongest interfering \acp{bs} per \ac{ue}, \note{ue}{n\in\N} selects \note{bs}{1} and \note{bs}{2} for cooperation, i.e., $\I_n'=\{1,2\}$. Thus, \note{ue}{n} generates ${J'=4}$ \mcsi{11} reports on \note{prb}{l\in\myL}, as summarized in \tablename~\ref{tab_csi_report}. According to the Proposition~\ref{lem_sinr} and the Corollary \ref{lem_rate}, $r_{n,l,2} \geq r_{n,l,3}, r_{n,l,4} \geq r_{n,l,1}$.
\end{ex}

\begin{table}[!t]
\renewcommand{\arraystretch}{1.3}
\caption{\mcsi{11} reports for \note{ue}{n}, on \note{prb}{l}, with $M'=2$}
\label{tab_csi_report}
\centering
\begin{tabular}{cccc}
\hline
\textbf{Int. Scenario} & \textbf{Mut. Ind.} & \textbf{Mut. Pattern} & \textbf{Achiev.}\\
$(j\in\J')$ & $(\J_{n,j})$ & $(\alpha_{n,l,j})$ & \textbf{Data rate}\\
\hline
1 & $\{\varnothing\}$ & $\left[0, 0, 0, 0\right]$ & $r_{n,l,1}$\\
2 & $\{1,2\}$ & $\left[1, 1, 0, 0\right]$ & $r_{n,l,2}$\\
3 & $\{1\}$ & $\left[1, 0, 0, 0\right]$ & $r_{n,l,3}$\\
4 & $\{2\}$ & $\left[0, 1, 0, 0\right]$ & $r_{n,l,4}$\\ \hline
\end{tabular}
\end{table}

\section{\ac{cs} with muting}\label{sec_cs}

\subsection{Proposed \ac{inlp} - Problem formulation}\label{subsec_general_ilp}
At the central controller, the \mcsi{11} reports generated by the \acp{ue} and forwarded by the \acp{bs}, are used in order to compute the \ac{cs} decision. The \ac{cs} decision consists of two main components, namely, a scheduling decision that assigns \acp{prb} to \acp{ue}, and a muting decision that mutes \acp{bs} on particular \acp{prb}, to reduce the interference experienced by the \acp{ue} connected to neighboring \acp{bs}. The matrix variable $\bar{\mathbf{S}}$ of dimensions $N\times L$ and elements
\begin{equation}
\label{eq_sched_dec}
	\bar{s}_{n,l}= \left\{
		\begin{aligned}
			1 & \ \text{if \ac{prb} }l\in\myL \text{ is assigned to \ac{ue} }n\in\N\\
			0 & \ \text{otherwise},\\
		\end{aligned}\right.
\end{equation}
is used to denote the scheduling decision for all \acp{ue} on each \note{prb}{l}, while the $M\times L$ matrix variable $\balfa{0}{}$, with elements as introduced in \eqref{eq_mute_dec}, refers to the muting decision for all \acp{bs} on each \note{prb}{l}. Both decisions depend on each other. On the one hand, the selection of the \acp{ue} to be served in a given \note{prb}{l}, depends on the data rates these \acp{ue} can achieve under a particular muting decision. On the other hand, the muting decision depends on the margin by which the achievable data rates of the \acp{ue} to be served increases with respect to the loss on the achievable data rates of the \acp{ue} connected to the muted \acp{bs}, for that particular muting decision. In the following, an \ac{inlp} is proposed, to carry out joint \ac{bs} muting and \ac{ue} scheduling in a coordinated network.

Typically, the schedulers in mobile communications pursue a trade-off between user throughput and fairness. For that purpose, opportunistic scheduling is applied such as in the case of the \ac{pf} scheduler \cite{Scheduling_Zhu,schedulers_Wu}. The objective of the \ac{pf} scheduler is to maximize the sum, over all \acp{ue}, of the \ac{pf} metrics given by
\begin{equation}
\label{eq_pfs}
	\Omega_n = \frac{r_n}{R_n} \ \forall n\in\N,
\end{equation}
where the ratio between the total instantaneous achievable data rate and the average user throughput over time, denoted by $r_n$ and $R_n$, respectively, of \note{ue}{n\in\N} is considered. The total instantaneous achievable data rate of \note{ue}{n} is calculated as
\begin{equation}
\label{eq_r_inst}
	r_n=g\left(r_{n,l,j},\boldsymbol{\bar{s}}_n,\balfa{0}{}\right) \ \forall l\in\myL, \forall j\in\J',
\end{equation}
where $g(\cdot)$ denotes a function of the achievable data rates of \note{ue}{n}, as defined in \eqref{eq_r_gen_2}, over the \acp{prb} assigned to \note{ue}{n}, as described by the $n$-\textit{th} row of $\bar{\mathbf{S}}$, denoted by $\boldsymbol{\bar{s}}_n$, and the muting decision matrix $\balfa{0}{}$.

The \ac{lte}-Advanced \ac{cs} with muting problem can be formulated as the following \ac{inlp}
\begin{subequations}
\label{eq_imultilp}
	\begin{alignat}{2}
		\underset{\left\{\mathbf{\bar{S}},\boldsymbol{\bar{\alpha}}\right\}}{\max} \,&&&\sum_{n\in\N}\Omega_n \label{eq_multil_obj}\\
		\text{s.t.} \,&&&\nonumber\\
		&&&\bar{\alpha}_{m,l} + \sum_{n\in\N}c_{n,m}\,\bar{s}_{n,l}\leq1 \ \forall m\in\M, \forall l\in\myL, \label{eq_multil_valid_alpha}\\
		&&&r_n = \sum_{l\in\myL}\rho\left(\boldsymbol{r}_{n,l},\balfa{1}{l},\I_n'\right)\,\bar{s}_{n,l}\ \forall n\in\N, \label{eq_multil_total_r}\\
		&&&\bar{s}_{n,l} \in \{0,1\} \ \forall n\in\N, \forall l\in\myL, \label{eq_multil_binary_S}\\
		&&&\bar{\alpha}_{m,l} \in \{0,1\} \ \forall m\in\M, \forall l\in\myL, \label{eq_multil_binary_alpha}
	\end{alignat}
\end{subequations}
where the objective in \eqref{eq_multil_obj} is to maximize the sum of the \ac{pf} metrics over all \acp{ue}, with the \ac{pf} metric of \note{ue}{n\in\N} calculated as in \eqref{eq_pfs}. The constraints in~\eqref{eq_multil_valid_alpha} link the scheduling decision $\mathbf{\bar{S}}$ with the muting decision $\boldsymbol{\bar{\alpha}}$. If \note{bs}{m\in\M} is muted on \note{prb}{l\in\myL}, then \note{prb}{l} should not be assigned to any \note{ue}{n} connected to \note{bs}{m}. Thus, if $\bar{\alpha}_{m,l}=1$ in \eqref{eq_multil_valid_alpha}, for \note{bs}{m}, the second term on the left-hand-side must be equal to zero, which is true in either of the following cases, with the connection indicator $c_{n,m}$ given by \eqref{eq_c}:
\begin{itemize}
	\item No \acp{ue} are connected to \note{bs}{m}, i.e., ${c_{n,m}=0}$, ${\forall n\in\N}$.
	\item \note{prb}{l} is not assigned to any \ac{ue} served by \note{bs}{m}, i.e., $\bar{s}_{n,l}=0$, ${\forall n\in\N}$ such that $c_{n,m}=1$.
\end{itemize}
Furthermore, in the case that \note{bs}{m} is not muted on \note{prb}{l}, i.e., $\bar{\alpha}_{m,l}=0$, the constraints in \eqref{eq_multil_valid_alpha} ensure that single user transmissions are carried out, where each \ac{bs} is allowed to schedule a maximum of one \ac{ue} per \ac{prb}. Additionally, the total instantaneous achievable data rate of \note{ue}{n}, denoted by $r_n$ as introduced in \eqref{eq_r_inst}, is calculated in \eqref{eq_multil_total_r}, with
\begin{equation}
\label{eq_g}
	g\left(r_{n,l,j},\boldsymbol{\bar{s}}_n,\balfa{0}{}\right) = \sum_{l\in\myL}\rho\left(\boldsymbol{r}_{n,l},\balfa{1}{l},\I_n'\right)\,\bar{s}_{n,l}\ \forall n\in\N.
\end{equation}
In \eqref{eq_g}, $\rho\left(\boldsymbol{r}_{n,l},\balfa{1}{l},\I_n'\right)$ is a lookup table function that selects the achievable data rate of \note{ue}{n}, on \note{prb}{l}, based on the muting decision $\balfa{1}{l}$ of the strongest interfering \acp{bs} of \note{ue}{n}, as given by $\I_n'$. The lookup table function, $\rho\left(\cdot\right)$, selects the achievable data rate from the $J'\times 1$ vector, $\boldsymbol{r}_{n,l}$, with elements $r_{n,l,j},~\forall j\in\J'$, obtained from the \mcsi{11} reports of \note{ue}{n}, on \note{prb}{l}, as defined in \eqref{eq_r_gen_2}.
\begin{ex}
Based on the \tablename~\ref{tab_csi_report} from Example~\ref{ex_csi}, with ${J'=4}$ \mcsi{11} reports, the lookup table function for \note{ue}{n\in\N}, on \note{prb}{l\in\myL}, provides the results as in \tablename~\ref{tab_rho}. Note that the value of $\rho\left(\cdot\right)$ does not depend on the muting decision of the remaining \acp{bs}.
\end{ex}

Due to the utilization of the $J'$ \mcsi{11} reports, the achievable data rate of \note{ue}{n\in\N}, on \note{prb}{l\in\myL}, under interference scenario $j\in\J'$, is constant in the problem formulation and limited to the set of reported muting patterns. Additionally, taking into account the above introduced lookup table function, $\rho\left(\cdot\right)$, and assuming that the achievable data rate function $f\left(\gamma_{n,l}\right)$, as defined in~\eqref{eq_r_gen}, is piece-wise non-decreasing, the following Proposition~\ref{lem_muting} applies.

\begin{prop}\label{lem_muting}
For \note{ue}{n\in\N}, on \note{prb}{l\in\myL}, if ${\J_{n,i}\subsetneq\J_{n,j}, \forall i,j\in\J', i\neq j}$ and $r_{n,l,i} = r_{n,l,j}$, then the interference scenario $i\in\J'$ provides the highest sum of \ac{pf} metrics over all \acp{ue} in the cooperation cluster, among both scenarios~$i,j\in\J'$.
\end{prop}
\begin{proof}
See Appendix \ref{app_lem_muting}.
\end{proof}

Hence, from Proposition \ref{lem_muting}, it follows that additional (strongest) interfering \acp{bs} are only muted if the achievable data rate of \note{ue}{n\in\N} is increased.

\begin{table}[!t]
\renewcommand{\arraystretch}{1.3}
\caption{Lookup table function $\rho\left(\boldsymbol{r}_{n,l},\balfa{1}{l},\I_n'\right)$ for \note{ue}{n}, on \note{prb}{l}, with $M'=2$}
\label{tab_rho}
\centering
\begin{tabular}{cc}
\hline
$\balfa{1}{m,l},~\forall m\in\I_n'$ & $\rho\left(\boldsymbol{r}_{n,l},\balfa{1}{l},\I_n'\right)$\\
\hline
$\left[0, 0\right]$ & $r_{n,l,1}$\\
$\left[0, 1\right]$ & $r_{n,l,4}$\\
$\left[1, 0\right]$ & $r_{n,l,3}$\\
$\left[1, 1\right]$ & $r_{n,l,2}$\\ \hline
\end{tabular}
\end{table}

Moreover, as previously explained, the scheduling and muting matrix variables $\mathbf{\bar{S}}$ and $\boldsymbol{\bar{\alpha}}$, are binary as described by the constraints in \eqref{eq_multil_binary_S} and \eqref{eq_multil_binary_alpha}, respectively.

The following remarks summarize the characteristics of the \ac{lte}-Advanced \ac{cs} with muting problem formulation in \eqref{eq_imultilp}.
\begin{itemize}
	\item As mentioned in Section \ref{sec_csi_rep}, given $M'$ strongest interfering \acp{bs} per each \note{ue}{n\in\N}, a total of ${J'=2^{M'}}$ interfering scenarios per \note{ue}{n} are available. Hence, two special cases of the problem formulation are observed:	
		\item[\textit{i)}] If $M'=0$, each \note{ue}{n} generates one \mcsi{11} report under the assumption of no \ac{bs} muting. At the central controller, the \ac{cs} with muting problem formulation becomes a \ac{pf} scheduler without any cooperation.
		\item[\textit{ii)}] If $M'=M-1$, all the interfering \acp{bs} within the cooperation cluster can be muted to improve the performance of any \ac{ue}, on each \note{prb}{l\in\myL}. If the network size is large, finding the solution while assuming cooperation of all interfering \acp{bs} for all \acp{ue} approximates an exhaustive search.
	\item The problem is purely integer, and furthermore binary because of the constraints in \eqref{eq_multil_binary_S} and \eqref{eq_multil_binary_alpha}.	
	\item Due to the combinatorial nature of the problem formulation, it is classified as non-deterministic polynomial-time (NP)-hard.
	\item The problem is non-linear because of the relation between the muting and the scheduling decision variables, $\balfa{1}{l}$ and $\bar{s}_{n,l}$, respectively, in the constraints in \eqref{eq_multil_total_r}.
\end{itemize}

Although the number of reported interference scenarios $J'=2^{M'}$ can be limited by selecting a small value $M'$ of (strongest) interfering \acp{bs} per \note{ue}{n\in\N}, the \ac{cs} with muting \ac{inlp} formulation in \eqref{eq_imultilp} also depends on the number of \acp{ue}, i.e., $N$, and the number of \acp{prb}, denoted by $L$. For certain network scenarios, $N$ and $L$ can be large. Therefore, given the non-linear nature of the problem in \eqref{eq_imultilp}, finding a solution with commercial solvers may either not be possible or inefficient in terms of computation time. In the following, \emph{separability}, \emph{reducibility} and \emph{lifting} concepts are used, in order to formulate parallel \ac{ilp} sub-problems that scale better with the network size.

\subsection{Proposed \ac{ilp} - Parallelized sub-problem formulation}\label{subsec_parallel_ilp}
\emph{Separability.} When analyzing the objective function described by \eqref{eq_multil_obj}, the total \ac{pf} metric corresponds to the sum of the individual \ac{pf} metrics for all \acp{ue}. Furthermore, at each \note{ue}{n\in\N}, it is assumed that the total instantaneous achievable data rate is equivalent to the linear combination of the decoupled achievable data rates per scheduled \acp{prb}, as given by \eqref{eq_multil_total_r}. Therefore, it is possible to separate the \ac{cs} with muting problem in \eqref{eq_imultilp}, into $L$ independent sub-problems, corresponding to the scheduling decision of one \ac{prb} each. By performing this parallelization, the computation time is reduced without affecting the quality of the solution, i.e., the solution of the parallelized \ac{cs} with muting problem remains optimal.

\emph{Reducibility.} It is expected that some of the \acp{ue} connected to a common \note{bs}{m\in\M}, share one or more strongest interfering \acp{bs}. From a \ac{bs} perspective, the set
\begin{equation}
\label{eq_unique_int}
	\J_m=\underset{n\in\N~|~c_{n,m}=1}{\cup}\J_n \ \forall m\in\M,
\end{equation}
contains all the unique muting indicator sets, associated to its connected \acp{ue}. Similar to the set $\J_n$, $\J_m$ is common to all \acp{prb} in the reporting period. The number of unique muting indicator sets for \note{bs}{m}, i.e., $J_m'=|\J_m|$, depends on the number of \acp{ue} connected to \note{bs}{m} and the maximum number $J'$ of reported interference scenarios per \ac{ue}, as introduced in Section~\ref{sec_csi_rep}. Thus, ${J'\leq J_m'\leq\sum_{n\in\N}c_{n,m}\,J'}$, where the lower bound corresponds to the case when all connected \acp{ue} are interfered by the same set of strongest interfering \acp{bs}, and the upper bound represents the case with all \acp{ue} having different strongest interfering \acp{bs}. For the unique muting indicator set $\J_{m,j'}$, with $j'\in\J_m'=\{1,\dots,J_m'\}$, the set of indexes of \acp{ue}, connected to \note{bs}{m}, with equal muting indicator set is defined as
\begin{equation}
\label{eq_ue_int}
	\N_{m,j'} = \{n\in\N~|~c_{n,m}=1,\J_{m,j'}\subsetneq\J_n,~\forall m\in\M, \forall j'\in\J_m'\}.
\end{equation}
Based on the definitions in \eqref{eq_unique_int} and \eqref{eq_ue_int}, the following Proposition~\ref{lem_reducibility} is given.

\begin{prop}\label{lem_reducibility}
For \note{bs}{m\in\M}, with unique muting indicator set index given by ${j'\in\J_m'}$, and the set $\N_{m,j'}$ as introduced in \eqref{eq_ue_int}. If $|\N_{m,j'}|>1$, then the optimal contribution of \note{bs}{m}, on \note{prb}{l\in\myL}, to the total \ac{pf} metric in \eqref{eq_multil_obj}, corresponds to the \ac{pf} metric of \note{ue}{\hat{n}}, with
\begin{equation}
\label{eq_lem_red}
	\hat{n} = \underset{n\in\N_{m,j'}}{\arg\max}\, \Omega_{n,l,j} \ \forall j\in\J'~|~\J_{n,j}=\J_{m,j'}.
\end{equation}
\end{prop}
\begin{proof}
See Appendix \ref{app_lem_reducibility}.
\end{proof}

Based on Proposition~\ref{lem_reducibility}, it is sufficient that each \note{bs}{m\in\M} forwards to the central controller, the \mcsi{11} reports related to one \ac{ue} per unique muting indicator set $\J_{m,j'}, \forall j'\in\J_m'$, on \note{prb}{l\in\myL}, instead of the \mcsi{11} reports from all connected \acp{ue}. The set of indexes of \acp{ue} connected to \note{bs}{m} that maximize the \ac{pf} metric, in at least one of the unique muting indicators sets indexed by~$j'\in\J_m'$, on \note{prb}{l}, is defined as 
\begin{equation}
\label{eq_red_N}
	\N_{m,l}'=\{\hat{n}~|~\exists j: \hat{n}=\underset{n\in\N_{m,j'}}{\arg\max}\,\Omega_{n,l,j},~\forall j'\in\J_m',\forall j\in\J'~|~\J_{n,j}=\J_{m,j'}\}.
\end{equation}
The cardinality of the set $\N_{m,l}'$, is bounded as ${1\leq|\N_{m,l}'|\leq\sum_{n\in\N}c_{n,m}}$, where the lower bound implies that only one \ac{ue} provides the maximum \ac{pf} metric, among all unique muting indicator sets on \note{prb}{l}, and the upper bound corresponds to the case when each \ac{ue} reports different muting indicator sets with respect to the other \acp{ue} connected to \note{bs}{m}.

At the central controller, all the achievable data rates, $r_{n,l,j}, \forall n\in\N_{m,l}'$, $\forall m\in\M, \forall l\in\myL, \forall j\in\J'$, are per definition set to zero, for the interference scenarios where \note{ue}{n} does not provide the maximum \ac{pf} metric, among the \acp{ue} connected to the same \note{bs}{m}. The set $\N_l'=\underset{\mathclap{m\in\M}}{\cup}\ \N_{m,l}'$ is used to denote the indexes of \acp{ue} to be considered in the reformulated \ac{ilp}, on \note{prb}{l}. The cardinality of the set $\N_l'$ is described as $M\leq|\N_l'|\leq N$. In the special case of $M'=0$, all \acp{ue} report only one interference scenario where no cooperative interfering \ac{bs} is muted, and thus, $|\N_l'|=M$.

\textit{Lifting.} In order to linearize the constraints in \eqref{eq_multil_total_r}, a variable transformation is introduced based on the lifting technique \cite{lifting_Balas}. A new coordinated decision variable is defined containing both, the scheduling and the muting decisions, as
\begin{equation}
\label{eq_coord_dec}
	s_{n,l,j}= \left\{
		\begin{aligned}
			1 & \ \text{if \ac{prb} }l\in\myL \text{ is assigned to \ac{ue} }n\in\N\\
			   & \ \text{under interference scenario }j\in\J'\\
			0 & \ \text{otherwise}.
		\end{aligned}\right.
\end{equation}
The new decision variable, $s_{n,l,j}$, is related to the muting and scheduling decisions in \eqref{eq_mute_dec} and \eqref{eq_sched_dec}, respectively, as
\begin{equation}
\label{eq_theo_s}
	s_{n,l,j} = 1 \Leftrightarrow \:\bar{s}_{n,l}=1 \land \bar{\alpha}_{m,l}=1~\forall n\in\N, \forall l\in\myL, \forall j\in\J', \forall m\in\J_{n,j},
\end{equation}
with $\land$ denoting the logical \emph{and} operator. Hence, the non-linear constraints in~\eqref{eq_multil_total_r} reduce to a linear combination of the achievable data rates, i.e., $r_{n,l,j}$, and the new decision variable, $s_{n,l,j}$.

\textit{Problem Reformulation.} Using the above described concepts of separability, reducibility and lifting, the \ac{cs} with muting \ac{inlp} formulation in \eqref{eq_imultilp} can be reformulated as an \ac{ilp}, which can be efficiently solved by commercial solvers. Hence, with the set $\N_l'$ and defining the binary decision variable $\mathbf{S}_l$ to have dimensions $|\N_l'| \times J'$, the sub-problem formulation for \note{prb}{l\in\myL} is
\begin{subequations}
\label{eq_par_ilp}
	\begin{alignat}{2}
		\underset{\{\mathbf{S}_l\}}{\max} \, &&&\sum_{n\in\N_l'} \Omega_{n,l} \label{eq_subp_obj}\\
		\text{s.t.} \, &&&\nonumber\\
		&&&\begin{aligned}
			&s_{n,l,j} + \sum_{k\in\N_l'}\sum_{i\in\J'} c_{k,m}\,s_{k,l,i} \leq 1\\
			& \quad \quad \forall n\in\N_l', \forall j\in\J', \forall m \in \J_{n,j},
		\end{aligned}\label{eq_subp_valid_j}\\
		&&&\begin{aligned}
			&\sum_{n\in\N_l'}\sum_{j\in\J'} c_{n,m}\,s_{n,l,j}\leq1\\
			& \quad \quad \forall m\in\M\backslash\cup_{n\in\N_l'}\I_n',
		\end{aligned}\label{eq_subp_m}\\
		&&&r_{n,l} = \sum_{j\in\J'} r_{n,l,j}\,s_{n,l,j} \ \forall n\in\N_l', \label{eq_subp_total_r}\\
		&&&s_{n,l,j}=0 \ \forall n\in\N_l', \forall j\in\J'~|~r_{n,l,j}=0, \label{eq_subp_set_var}\\
		&&&s_{n,l,j} \in \{0,1\} \ \forall n\in\N_l', \forall j\in\J', \label{eq_subp_binary_S}
	\end{alignat}
\end{subequations}
where the objective in \eqref{eq_subp_obj} is to maximize the sum of the \ac{pf} metric over all \acp{ue}. The constraints in \eqref{eq_subp_valid_j} restrict the scheduling decisions of the strongest interfering \acp{bs} of \note{ue}{n\in\N_l'}, i.e., $\forall m\in\J_{n,j}$, in order to agree with the muting state considered in the interference scenario $j\in\J'$. If \note{prb}{l} is assigned to \note{ue}{n}, under the condition of muting the (strongest) interfering \acp{bs} indexed by the set $\J_{n,j}\in\J_n$, then no other \ac{ue} connected to the muted \acp{bs} can be simultaneously scheduled on the same \note{prb}{l}. Thus, if $s_{n,l,j}=1$ in \eqref{eq_subp_valid_j}, the second term on the left-hand-side must be equal to zero. Furthermore, in the case that $s_{n,l,j}=0$, the constraints in \eqref{eq_subp_valid_j} ensure that single user transmissions are carried out, where each \note{bs}{m\in\J_{n,j}} is allowed to schedule a maximum of one \ac{ue} per \ac{prb}, over all possible interference scenarios $j\in\J'$. Since it is possible that specific \acp{bs}, within the cooperation cluster, do not belong to the set of strongest interfering \acp{bs} of any \ac{ue}, the constraints in \eqref{eq_subp_m} complement the restriction on the single user transmissions from \eqref{eq_subp_valid_j}. Additionally, the total instantaneous achievable data rate of \note{ue}{n}, on \note{prb}{l}, denoted by $r_{n,l}$, is calculated in \eqref{eq_subp_total_r} as the achievable data rate for the selected interference scenario~$j$, as defined by the coordinated decision variable $s_{n,l,j}$. It is worth to notice that there is a one-to-one mapping between $r_{n,l,j}$ and $s_{n,l,j}$, thus, there is no requirement for a lookup table function as used in \eqref{eq_g}. Furthermore, the constraints in \eqref{eq_subp_set_var} are incorporated as a preprocessing step to ensure that no \ac{prb} is scheduled to \acp{ue} for which a maximum \ac{pf} metric for the corresponding interference scenario $j$ is not available. Finally, the coordinated decision variable $\mathbf{S}_l$ is binary as described by the constraints in \eqref{eq_subp_binary_S}.

It can be easily proven that the problem formulations in \eqref{eq_imultilp} and \eqref{eq_par_ilp} are equivalent. Furthermore, the proposed parallelized formulation in \eqref{eq_par_ilp}, reduces significantly the \ac{cs} with muting problem complexity, allowing its application even for large-size networks as illustrated in Section~\ref{sec_results}.

\subsection{Generalized greedy heuristic algorithm}\label{subsec_gen_heur}
The greedy heuristic deflation algorithm in \cite{Nokia_2015} (see algorithm in Section II), iteratively solves the \ac{cs} with muting problem per \note{prb}{l\in\myL}, where at each iteration, one \ac{bs} is muted, corresponding to the \note{bs}{m\in\M} that, when muted, maximizes the sum of the \ac{pf} metrics among all \acp{ue} on \note{prb}{l}. The algorithm stops when muting any additional \ac{bs} does not improve the sum of the \ac{pf} metrics with respect to the previous iteration. There is no guarantee that the heuristic algorithm yields a globally optimal point, because the quality of the scheduling decision depends directly on the gain achieved from muting one interfering \ac{bs} at a time.

Given the above mentioned disadvantage of the \ac{cs} with muting greedy heuristic algorithm from \cite{Nokia_2015}, an extension is proposed in this work, called \emph{generalized greedy heuristic algorithm}, which trades off computational complexity with performance gains. The main difference with respect to the algorithm in \cite{Nokia_2015}, is the evaluation of additional muting patterns per iteration, where for \note{prb}{l\in\myL}, the set of muting indicators
\begin{equation}
	\label{eq_muting_pos}
		\hat{\M}=\bigcup_{\hat{m}\in\{1,\dots,\tilde{m}\}}{\binom{\M'}{\hat{m}}},
	\end{equation}
defines the muting patterns to be evaluated. In \eqref{eq_muting_pos}, the binomial coefficients of the set $\M'$, of possible muted \acp{bs}, are evaluated by selecting $\hat{m}$ \acp{bs} at a time. The configuration parameter ${1\leq\tilde{m}\leq M-1}$, controls the complexity of the proposed generalized greedy heuristic algorithm by determining the muting patterns to be evaluated. If $\tilde{m}=1$, the generalized greedy heuristic algorithm reduces to the heuristic algorithm from \cite{Nokia_2015}. In the case that $\tilde{m}=M-1$, the generalized greedy heuristic algorithm performs an exhaustive search.

\section{Simulation results}\label{sec_results}
In this section, extensive simulation results are presented to evaluate the performance of the \ac{comp} \ac{cs} schemes with respect to a \ac{pf} scheduler without any cooperation, referred to as ``non-coop. PFS". The proposed parallelized sub-problem formulation as presented in Section \ref{subsec_parallel_ilp}, labeled as ``CS-ILP", is examined, together with the greedy algorithm described in \cite{Nokia_2015}, denoted as ``CS-GA", and the proposed generalized greedy algorithm of Section \ref{subsec_gen_heur}, labeled as ``CS-GG". In the simulations, ${M'=2}$ strongest interfering \acp{bs} per \ac{ue} are considered.

\subsection{\ac{cs} with muting - Performance analysis}\label{subsec_ana_results}
In order to study the performance of the \ac{cs} with muting schemes, Monte Carlo standalone simulations have been carried out, where the \mcsi{11} reports are generated based on channels obtained from a \ac{3gpp} compliant system level simulator, as specified in \cite{3GPP_TR_25.996,3GPP_TR_36.814,3GPP_TR_36.872,3GPP_TR_36.874}. In each transmission time interval $t$, the average user throughput over time of \note{ue}{n\in\N}, used in \eqref{eq_pfs}, is updated based on the scheduling decisions made at the previous transmission time interval $t-1$, as
\begin{equation}
\label{eq_R}
	R_n(t) = \beta\,R_n(t-1)+ (1-\beta)\,r_n(t-1),
\end{equation}
with $\beta=0.97$, denoting the forgetting factor parameter used to trade-off user throughput and fairness \cite{average_R_Motorola}. The total instantaneous achievable data rate of \note{ue}{n}, at the previous transmission time interval, denoted by $r_n(t-1)$, is calculated as given by e.g., \eqref{eq_subp_total_r}.

Initially, the performance of the \ac{cs} with muting  algorithms, in terms of average user throughput, is studied with respect to the data rates the users can achieve per symbol and to the noise power level considered in the calculation of these achievable data rates. In practical systems such as \ac{lte}-Advanced, finite \acp{mcs} are used which restrict the achievable data rates per symbol to a given range \cite{Sesia_LTE,CQI_table_Yong}. For the current analysis, two cases are considered with respect to the maximum achievable data rate: \textit{i)} the \ac{mcs} is unbounded, denoted as ``Unb. MCS", where the maximum achievable data rate can approach arbitrarily large values, and \textit{ii)} a maximum achievable data rate of \SI{5.4}{bits/symbol} is used, as imposed by a typical highest \ac{mcs} bound in \ac{lte}-Advanced, referred to as ``B. MCS". Similarly, there are two assumptions with respect to the noise power level, where in a first case, noise free decoding is assumed, denoted as ``N.-less", which considers that $\sigma^2=\epsilon$, with $\epsilon$ arbitrarily small but larger than zero, and in a second, a typical receiver noise figure of \SI{9}{dB} is considered, referred to as the ``Noisy" case.

The cell-edge and the geometric mean of the user throughput are shown in \figurename s~\ref{fig_ce_thr_sens} and \ref{fig_geomean_thr_sens}, respectively, for a scenario with $M=3$ \acp{bs}, $N=30$ \acp{ue} (10 \acp{ue} per \ac{bs}) and $L=10$ \acp{prb}. The cell-edge throughput describes the average user throughput of the cell-edge users and corresponds to the average throughput achieved by the worst \SI{5}{\%} of the users. The usage of the geometric mean is proposed by the authors in \cite{Nokia_2015} as a direct measure of the \ac{pf} scheduler's objective function. The user throughputs achieved by the \ac{cs} with muting schemes, i.e., CS-ILP, CS-GA and CS-GG, are normalized by the resulting user throughput when no cooperative scheduler is applied, i.e., non-coop. PFS. Four cases are considered for different combinations of maximum achievable data rate and noise power level, as specified in the horizontal axis. No additional \acp{bs} are considered in the network, hence, there is no out-of-cluster interference, i.e., $I_{n,l}^{\text{oc}}=0$. It is observed that under no achievable data rate limitations, and noise free receivers, i.e., Unb. MCS and N.-less, significant user throughput gains for both, the cell-edge and the geometric mean, are achieved by the cooperative schemes, with respect to the non-coop. PFS. Moreover, the optimality of the proposed CS-ILP formulation is notable, with the CS-GA being unable to obtain the optimal solution as explained in Section \ref{subsec_gen_heur}. Due to the unboundedness of the \ac{mcs} and the noise free decoder assumptions in this case, simultaneously muting the two interfering \acp{bs} can significantly increase the \ac{ue}'s data rate. Nevertheless, only muting one interfering \ac{bs} does not yield sufficient \ac{pf} metric gain, causing the CS-GA scheme to stop prematurely. Such a limitation of the CS-GA is not present in the proposed CS-GG, which achieves the same optimal performance as the CS-ILP scheme. Once limitations are assumed in either the maximum achievable data rate, or the noise power level, or both, the observed gains from the \ac{cs} with muting schemes, with respect to the non-coop. PFS approach, vanish. Due to the low number of \acp{bs} in the cooperation cluster and given the above mentioned limitations, few users benefit from the simultaneous muting of the two interfering \acp{bs}. Thus, a greedy algorithm performs almost optimal under such practical network assumptions. The average percentage of muted \acp{prb} per \ac{bs}, for the four different scheduling schemes and four combinations of maximum achievable data rate and noise power level, is presented in Table~\ref{tab_muted_prb}. The non-coop.~PFS does not apply muting, therefore the table contains zero entries for all cases. For the \ac{cs} with muting schemes, according to \figurename s~\ref{fig_ce_thr_sens} and \ref{fig_geomean_thr_sens}, the average percentage of  muted resources per \ac{bs} reduces when the gain of muting is restricted. It is worth to notice that even when the maximum achievable data rate is assumed to be unbounded above, i.e., Unb. MCS, and noiseless receivers are considered, i.e., N.-less, the CS-ILP scheme mutes $2/3$ of the resources per \ac{bs}, which means that each \ac{bs} orthogonally schedules its \acp{ue} over $1/M\text{-\emph{th}}$ of the available resources. Further muting resources per \ac{bs}, reduces the network performance because the user throughput distribution lacks fairness among the \acp{bs}. The value $1/M$, represents a fundamental limit of the cooperation and agrees with analytical studies presented by Lozano et al. in \cite{Comp_limit_Lozano}. Although the performance of the heuristic \ac{cs} with muting schemes is close-to-optimal under current practical network conditions, it is envisioned that the evolution of mobile communications introduces for future networks receivers with enhanced capabilities to suppress noise and to support the usage of higher \acp{mcs}. Hence, the results in \figurename s~\ref{fig_ce_thr_sens} and \ref{fig_geomean_thr_sens}, provide a reference to the potential gains of these heuristic schemes with respect to the optimal performance obtained with the proposed CS-ILP.

\begin{figure}[!t]
\centering
\includegraphics{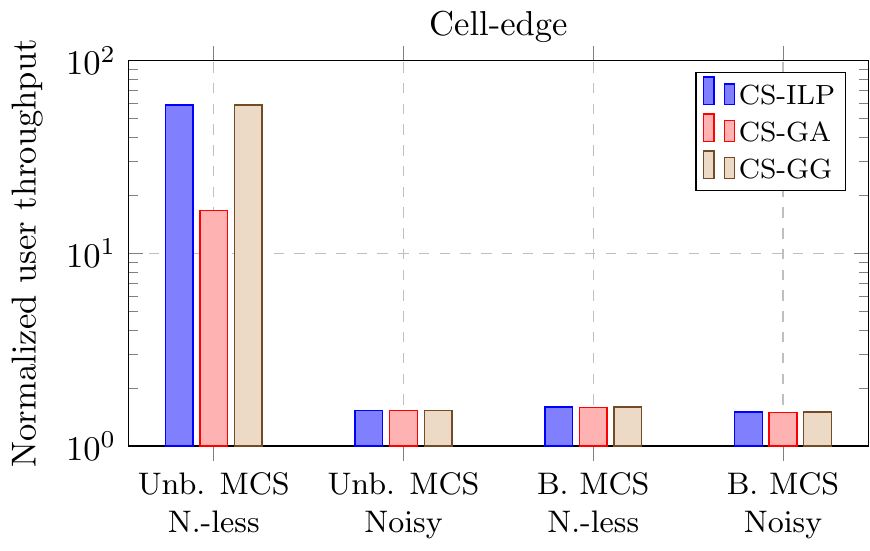}
\caption{Average cell-edge user throughput for the \ac{comp} \ac{cs} schemes, normalized with respect to the non-coop. PFS. Scenario with $M=3$ \acp{bs}, $N=30$ \acp{ue}, $L=10$ \acp{prb} and $M'=2$ \acp{bs}. Four cases with limitations on the maximum achievable data rate and the noise power level are considered. There is no out-of-cluster interference.}
\label{fig_ce_thr_sens}
\end{figure}

\begin{figure}[!t]
\centering
\includegraphics{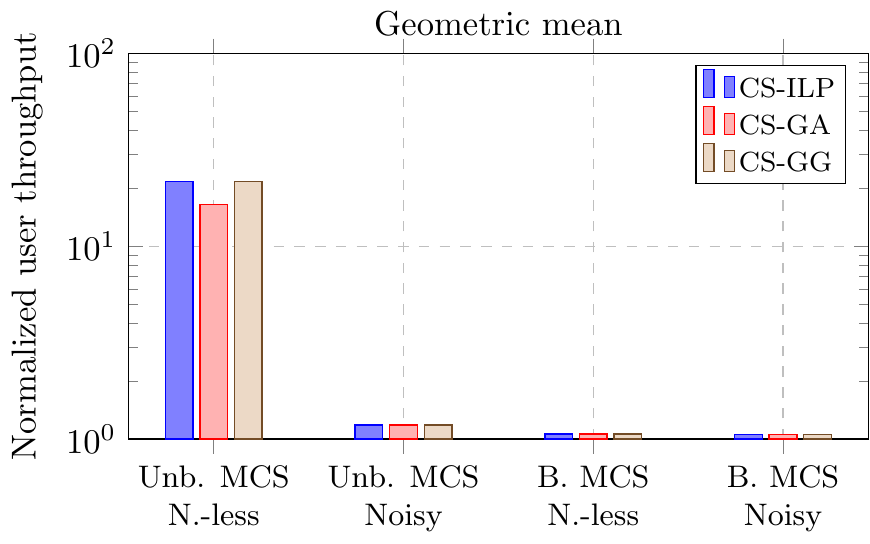}
\caption{Geometric mean of the average user throughput for the \ac{comp} \ac{cs} schemes, normalized with respect to the non-coop. PFS. Scenario with $M=3$ \acp{bs}, $N=30$ \acp{ue}, $L=10$ \acp{prb} and $M'=2$ \acp{bs}. Four cases with limitations on the maximum achievable data rate and the noise power level are considered. There is no out-of-cluster interference.}
\label{fig_geomean_thr_sens}
\end{figure}

\begin{table}[!t]
\renewcommand{\arraystretch}{1.3}
\caption{Average percentage of muted resources per \ac{bs}}
\label{tab_muted_prb}
\centering
\begin{tabular}{ccccc}
\hline
\textbf{Scheduling} & non-coop. & CS & CS & CS\\
\textbf{scheme} & PFS & ILP & GA & GG\\
\hline
Unb. MCS \& N.-less & 0 & 0.67 & 0.53 & 0.67\\
Unb. MCS \& Noisy & 0 & 0.22 & 0.21 & 0.22\\
B. MCS \& N.-less & 0 & 0.08 & 0.08 & 0.08\\
B. MCS \& Noisy & 0 & 0.08 & 0.07 & 0.08\\ \hline
\end{tabular}
\end{table}

In the following, the more practical scenario with bounded \ac{mcs} and noisy receivers, i.e., a maximum achievable data rate of \SI{5.4}{bits/symbol} and a noise figure of \SI{9}{dB}, is considered in order to study the performance of the \ac{cs} with muting schemes, with respect to the cooperation cluster size. For that purpose, a network of seven \acp{bs} is simulated, where a single cooperation cluster of variable size, with $M \in \{3,\dots,7\}$, is assumed. The \acp{bs} outside of the cooperation cluster are assumed to transmit data with maximum transmit power over the complete simulation time, i.e., ${I_{n,l}^{\text{oc}}\geq0}$. Additionally, two alternatives for the number of strongest interfering \acp{bs} per \ac{ue}, denoted by $M'$, are considered with $M'=M-1$ and $M'=2$. In the latter case, in order to have a conservative estimation of the achievable data rate, the \acp{ue} generate \mcsi{11} reports assuming maximum interference from the weakest interfering \acp{bs} as defined in~\eqref{eq_wi_int}. Each \ac{bs} serves 10 \acp{ue} over $L=10$ \acp{prb}. The cell-edge throughput, as a function of the cooperation cluster size $M$, is shown in \figurename~\ref{fig_cell_edge_Int} for the \acp{ue} served by the \acp{bs} within the cooperation cluster. The presented results are normalized with respect to the user throughput achieved by the same \acp{ue}, if the non-coop. PFS is used. In accordance to the previous results, the \ac{cs} with muting schemes provide gains with respect to a non-cooperative \ac{pf} scheduler, with an increase in the gains for a larger cooperation cluster size. The reason for such an improvement is the opportunity of further reducing the interference, and thus enhancing the \ac{sinr}, by increasing the amount of \acp{bs} involved in the coordinated scheduling procedures. It is also observable that a larger number of $M'$ strongest interfering \acp{bs} per \ac{ue} improves the gains of the \ac{cs} with muting schemes, at the cost of additional computational complexity and signaling overhead. In agreement with the results presented in \figurename s~\ref{fig_ce_thr_sens} and~\ref{fig_geomean_thr_sens}, the greedy algorithm of \cite{Nokia_2015} shows a close-to-optimal, i.e., close to CS-ILP, performance under practical conditions, with the proposed CS-GG algorithm performing better than the CS-GA scheme when all possible strongest interfering \acp{bs} are considered. Similar results were observed for the geometric mean of the user throughput.

\begin{figure}[!t]
\centering
\includegraphics{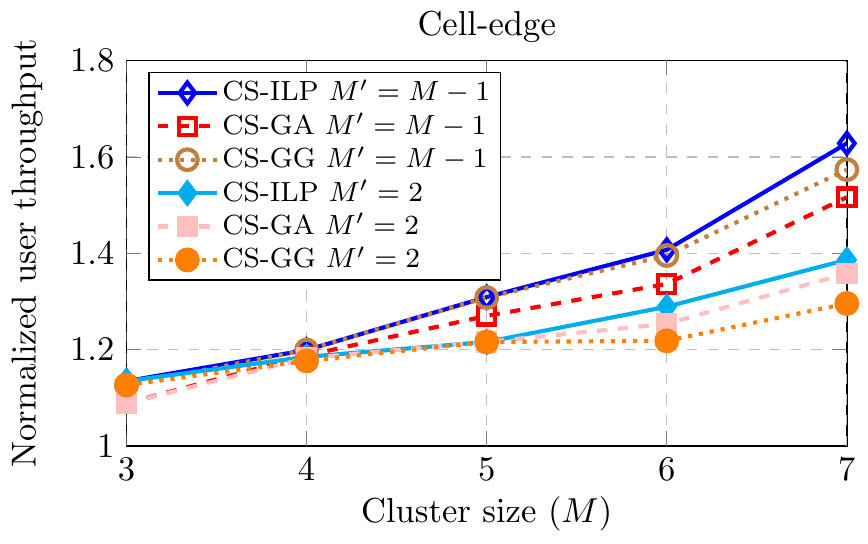}
\caption{Average cell-edge user throughput, normalized with respect to the non-coop. PFS, for different cooperation cluster sizes $(M)$. Scenario with ${N=10M}$ \acp{ue}, $L=10$ \acp{prb} and $M'=\{2,M-1\}$ \acp{bs}. There is out-of-cluster interference.}
\label{fig_cell_edge_Int}
\end{figure}

\subsection{\ac{cs} with muting - Potential gains}\label{subsec_ana_sls}
In this section, system level simulation results are presented in order to demonstrate the achievable gains of the \ac{cs} with muting schemes for \ac{lte}-Advanced macro-only and heterogeneous networks in an urban deployment. In both cases, $N=630$ \acp{ue} are served over $L=10$ \acp{prb}, by $M=21$ \acp{bs} in the macro-only network and $M=42$ \acp{bs} in the heterogeneous case where, one pico cell is located within the coverage area of a macro \ac{bs} with a separation distance of \SI{125}{m} from the macro \ac{bs}. The \acp{ue} are uniformly distributed in the macro-only case, while in the heterogeneous network the \acp{ue} are located in a hotspot fashion, where $2/3$ of the \acp{ue} are deployed in the vicinity of the pico \acp{bs}. As explained in Section \ref{sec_sys_mod}, in the heterogeneous networks cell range expansion is used with a \ac{sinr} off-set of \SI{6}{dB} for the small cells. The out-of-cluster interference is modeled using the wrap-around technique \cite{wrap_around_Yoon}, where additional \acp{bs} are deployed surrounding the $M$ \acp{bs} of interest. Additionally, \mcsi{11} reporting with periodicity of \SI{5}{ms} is applied where, similar to the simulations in Section~\ref{subsec_ana_results}, a conservative estimation of the achievable data rates is calculated by assuming maximum interference from the remaining \acp{bs}. Full buffer conditions, ideal link adaptation and rank one transmissions are assumed, i.e., all users are always active and demand as much data as possible, there are no decoding errors and only transmit beamforming is applied, respectively. For more information on \ac{3gpp}-compliant system level simulations, including channel and path-loss models, the interested reader is referred to \cite{3GPP_TR_36.814} (See 3GPP Case 1 and Case 6.2 from Section A.2.1).

\begin{figure}[!t]
\centering
\includegraphics{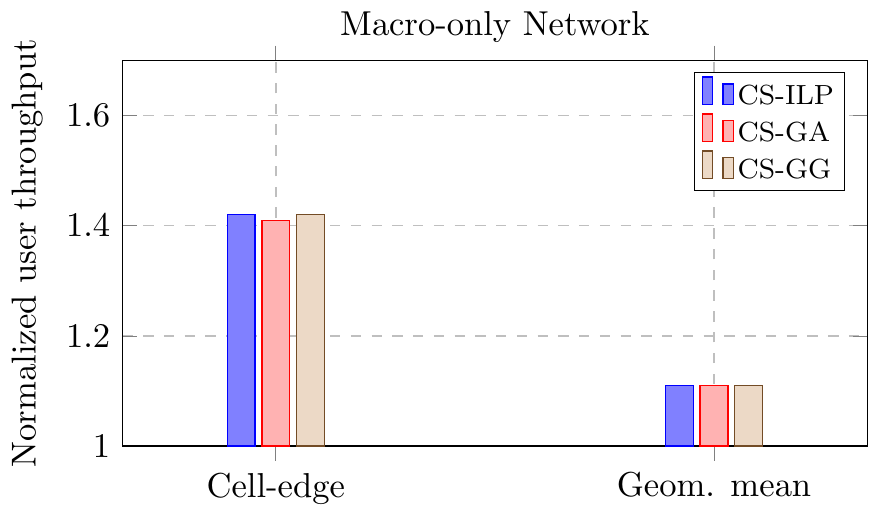}
\caption{Cell-edge and geometric mean of the average user throughput, normalized with respect to the non-coop. PFS, for a scenario with $M=21$ \acp{bs}, $N=630$ \acp{ue}, $L=10$ \acp{prb} and $M'=2$ \acp{bs}, with wrap-around technique. Results from system level simulations of a macro-only network.}
\label{fig_macro_only}
\end{figure}

\begin{figure}[!t]
\centering
\includegraphics{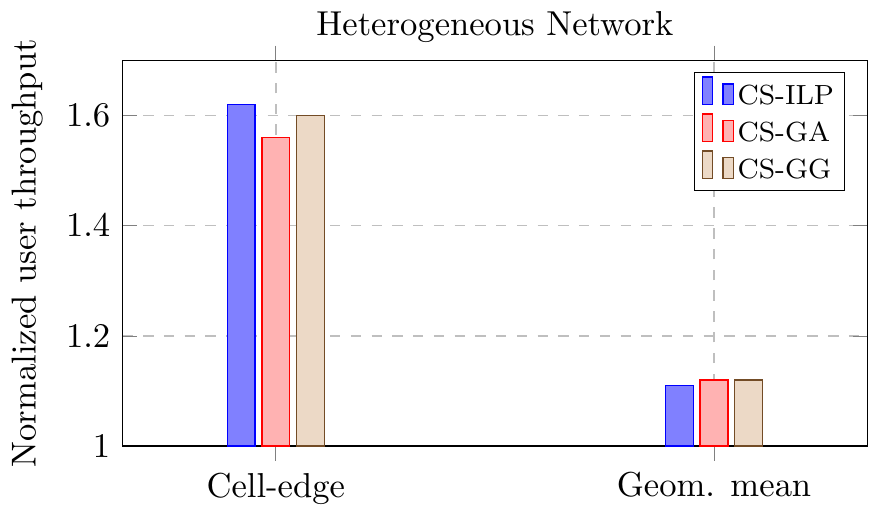}
\caption{Cell-edge and geometric mean of the average user throughput, normalized with respect to the non-coop. PFS, for a scenario with $M=42$ \acp{bs}, $N=630$ \acp{ue}, $L=10$ \acp{prb} and $M'=2$ \acp{bs}, with wrap-around technique. Results from system level simulations of a heterogeneous network.}
\label{fig_hetnet}
\end{figure}

The cell-edge and the geometric mean of the user throughput, normalized with respect to the non-coop. PFS, are presented in \figurename~\ref{fig_macro_only} for a macro-only, and in \figurename~\ref{fig_hetnet} for a heterogeneous network. In order to follow the standard \mcsi{11} reporting procedure, only $M'=2$ cooperative interfering \acp{bs} within the cooperation cluster are reported by each \ac{ue}. In terms of the geometric mean, gains are limited to values around \SI{11}{\%} for both cases, macro-only and heterogeneous networks. Additionally, the difference between the proposed schemes, i.e., CS-ILP and CS-GG, and the state-of-the-art CS-GA is negligible. For the \acp{ue} with the worst average user throughput, i.e., the cell-edge users, even with the limitation in the number of strongest interfering \acp{bs}, the \ac{cs} with muting schemes achieve a considerable gain in performance, with gains above \SI{40}{\%} being observable. In the case of heterogeneous networks, the cell-edge gain is even higher, due to the presence of a clear strongest interfering \ac{bs} for the pico \acp{ue}, i.e.,  the macro \ac{bs}, which is considered to cooperate within the restriction of $M'=2$. The proposed generalized greedy algorithm, i.e., CS-GG, performs better than the scheme in \cite{Nokia_2015}, i.e., CS-GA, which follows from the flexibility to muting additional \acp{bs}. The average percentages of muted \acp{prb} for the \ac{cs} with muting schemes in the macro-only and heterogeneous networks are presented in Table~\ref{tab_muted_prb_sls}. One implication of the muted \acp{prb} is the opportunity to save transmit power at the \acp{bs}, with the proposed CS-ILP and CS-GG schemes muting more \acp{prb} than the CS-GA scheme.

Finally, focusing on the proposed parallelized CS-ILP, it is recognizable that the simplifications proposed in Section \ref{subsec_parallel_ilp}, enable the implementation of such a \ac{cs} with muting approach even for medium to large-size networks. Hence, instead of solving the \ac{cs} with muting problem by considering the total of ${N=630}$ \acp{ue} per \note{prb}{l\in\myL}, only $|\N_l'|=136$ and ${|\N_l'|=213}$ \acp{ue} were included in average for the macro-only and the heterogeneous network, respectively. That implies a reduction of \SI{78}{\%} and \SI{66}{\%} in the problem size, for each of the cases, respectively.

\begin{table}[!t]
\renewcommand{\arraystretch}{1.3}
\caption{Average percentage of muted resources}
\label{tab_muted_prb_sls}
\centering
\begin{tabular}{cccc}
\hline
\textbf{Network} & CS-ILP & CS-GA & CS-GG\\
\hline
Macro-only & 0.11 & 0.10 & 0.10\\
Heterogeneous & 0.13 & 0.08 & 0.09\\ \hline
\end{tabular}
\end{table}

\section{Conclusions}
In this paper the coordinated scheduling with muting problem in the framework of \ac{lte}-Advanced networks with a centralized controller has been studied. A novel integer non-linear program formulation has been proposed to solve the problem optimally, where a computationally efficient equivalent integer linear program reformulation has been proposed to extend the applicability of the proposed scheme even to large-size networks.

Extensive system level simulation results show that coordinated scheduling with muting can potentially improve the cell-edge user performance, with higher gains in heterogeneous networks. Nevertheless, these gains are limited by the remaining uncoordinated interference and the finite time/frequency/space resources to be shared in the network.

The evaluation of the proposed integer linear program formulation, as well as the state-of-the-art heuristic greedy algorithm, for alternative traffic models in the non-full buffer case, are recommended for future studies. In the case of low demand, the possibility of reducing residual interference and increasing the degrees of freedom for the cooperation, can further enhance the performance gains of the mentioned coordinated scheduling schemes.

\section{Appendix}

\subsection{Proof of Proposition \ref{lem_sinr}}\label{app_lem_sinr}
Given the condition that ${\J_{n,i}\subsetneq\J_{n,j}, \forall i,j\in\J', i\neq j}$, the common strongest interfering \acp{bs} of \note{ue}{n\in\N} are considered to be muted in the interference scenarios $i$ and $j$. Thus, from the definition of the muting patterns in \eqref{eq_alpha}, ${\alpha_{n,m,l,i}=\alpha_{n,m,l,j}=1, \forall m\in\J_{n,i}}$. Furthermore, interference scenario $j$ mutes additional strongest interfering \acp{bs} in comparison to interference scenario $i$, i.e., ${\alpha_{n,m,l,i}=0,~ \alpha_{n,m,l,j}=1}$, ${\forall m\in\J_{n,j}\backslash\J_{n,i}}$. Thus, from \eqref{eq_si_int},
\begin{equation}
\label{eq_lem_sinr}
	I_{n,l,i}^{\text{si}}\left(\alpha_{n,m,l,i}\right)>I_{n,l,j}^{\text{si}}\left(\alpha_{n,m,l,j}\right).
\end{equation}
In \eqref{eq_sinr_2}, the interference from the strongest interfering \acp{bs} of \note{ue}{n} is the only term depending on interference scenarios $i$ and $j$. Therefore, taking into account the inequality in \eqref{eq_lem_sinr}, the \ac{sinr} of \note{ue}{n} on \note{prb}{l\in\myL}, under interference scenario~$j$ is higher.

\subsection{Proof of Proposition \ref{lem_muting}}\label{app_lem_muting}
The set ${\N_{n,j}=\{k~|~c_{k,m}=1, \forall k\in\N, \forall m\in\J_{n,j}\}}$ is defined, denoting the indexes of \acp{ue} connected to the strongest interfering \acp{bs} of \note{ue}{n\in\N} for interference scenario $j\in\J'$. Given that ${\J_{n,i}\subsetneq\J_{n,j}}$, then ${\N_{n,i}\subsetneq\N_{n,j}, \forall i,j\in\J', i\neq j}$. Based on \eqref{eq_multil_obj} and \eqref{eq_multil_total_r}, the sum of the \ac{pf} metrics over all \acp{ue} on \note{prb}{l\in\myL}, under interference scenario $y\in\J'$ of \note{ue}{n}, can be written as
\begin{equation}
\label{eq_lem_muting}
\sum_{n'\in\N}\Omega_{n',l}^y = \Omega_{n,l}^y +\ \sum_{\mathclap{k\in\N_{n,y}}}\Omega_{k,l} +\qquad \sum_{\mathclap{\hat{n}\in\N\backslash\{n,\, \N_{n,y}\}}}\, \Omega_{\hat{n},l},
\end{equation}
where the first right-hand-side summand corresponds to the \ac{pf} metric of \note{ue}{n} on \note{prb}{l}, under interference scenario $y$. The second summand corresponds to the sum of the \ac{pf} metrics of the \acp{ue} connected to the strongest interfering \acp{bs} of \note{ue}{n}, considered to be muted in the interference scenario $y$, and the last summand represents the sum of the \ac{pf} metrics of the \acp{ue} connected to the remaining \acp{bs}. If it is assumed that the muting decision agrees with interference scenario $y$, then the second summand is equal to zero, because the cooperative interfering \acp{bs} are muted. Thus, for interference scenarios $i$ and $j$, agreeing with the muting decision $\balfa{1}{l}$, \eqref{eq_lem_muting} is rewritten as
\begin{subequations}
\label{eq_lem_rate_ij}
	\begin{alignat}{1}
\sum_{n'\in\N}\Omega_{n',l}^i &= \Omega_{n,l}^i +\qquad \sum_{\mathclap{\hat{n}\in\N\backslash\{n,\, \N_{n,i}\}}}\, \Omega_{\hat{n},l}, \label{eq_lem_rate_i}\\
\sum_{n'\in\N}\Omega_{n',l}^j &= \Omega_{n,l}^j +\qquad \sum_{\mathclap{\hat{n}\in\N\backslash\{n,\, \N_{n,j}\}}}\, \Omega_{\hat{n},l}. \label{eq_lem_rate_j}
	\end{alignat}
\end{subequations}
If $r_{n,l,i} = r_{n,l,j}$, then ${\Omega_{n,l}^i = \Omega_{n,l}^j}$. Hence, the only difference between \eqref{eq_lem_rate_i} and~\eqref{eq_lem_rate_j} lays on the second summand. This summand is determined by the sets ${\N\backslash\{n,\N_{n,j}\}\subsetneq\N\backslash\{n,\N_{n,i}\}}$ due to $\N_{n,i}\subsetneq\N_{n,j}$, ${\forall i,j\in\J', i\neq j}$. Therefore, it is possible to conclude that, $\sum_{\hat{n}\in\N\backslash\{n,\, \N_{n,i}\}} \Omega_{\hat{n},l} > \sum_{\hat{n}\in\N\backslash\{n,\, \N_{n,j}\}} \Omega_{\hat{n},l}$ and thus,
\begin{equation}
\label{eq_lem_rate_proof}
\sum_{n'\in\N}\Omega_{n',l}^i>\sum_{n'\in\N}\Omega_{n',l}^j,
\end{equation}
where it has been assumed that each non-muted \ac{bs} schedules one \ac{ue} with a non-zero \ac{pf} metric.

\subsection{Proof of Proposition \ref{lem_reducibility}}\label{app_lem_reducibility}
It is assumed, without loss of generality, that $J'$ \mcsi{11} reports are generated by \acp{ue} $\{n,k\}\in\N$ and received by \note{bs}{m\in\M}, with equal muting indicators sets indexed by $\{j,i\}\in\J'$, respectively, such that ${\J_{n,j}=\J_{k,i}}$. Hence, from a \ac{bs} perspective, the unique muting indicator set $j'\in\J_m'~|~{\J_{m,j'}=\J_{n,j}=\J_{k,i}}$, implies that $\N_{m,j'}=\{n,k\}$. Based on \eqref{eq_pfs}, the \ac{pf} metrics of \acp{ue} $n$ and $k$, on \note{prb}{l\in\myL}, under unique muting indicator set $\J_{m,j'}$, correspond to 
\begin{equation}
\label{eq_proof_lem_red}
	\begin{aligned}
		\Omega_{n,l,j}=\frac{r_{n,l,j}}{R_n},\\
		\Omega_{k,l,i}=\frac{r_{k,l,i}}{R_k}.
	\end{aligned}
\end{equation}
Thus, the following relations are possible between the \ac{pf} metrics from \eqref{eq_proof_lem_red}: ${\Omega_{n,l,j}=\Omega_{k,l,i}}$, $\Omega_{n,l,j}<\Omega_{k,l,i}$ or ${\Omega_{n,l,j}>\Omega_{k,l,i}}$. In the first case, there is no effect on the total sum of \ac{pf} metrics if \note{bs}{m} schedules \note{prb}{l} to any of the both \acp{ue}, since the \ac{pf} metrics are equal. In the remaining cases, however, selecting the \ac{ue} with the lowest \ac{pf} metric corresponds to a lower total sum of the \ac{pf} metrics. Hence, the optimal allocation of \note{prb}{l} under muting indicator set $\J_{m,j'}$ is given by \eqref{eq_lem_red}.

\bibliographystyle{IEEEtran}
\bibliography{refs_article}

\end{document}